\documentclass[10pt,a4paper]{amsart}
\usepackage[utf8]{inputenc}
\usepackage{amsmath,amsfonts}

\usepackage{enumerate}
\usepackage{mathrsfs}
\usepackage{amssymb}
\usepackage{graphicx}

\usepackage[]{algorithm2e}
\usepackage{algorithmicx}

  \newtheorem{theorem}{Theorem}[section]

  \newtheorem{lemma}[theorem]{Lemma}

  \newtheorem{remark}[theorem]{Remark}
  \numberwithin{equation}{section}

\def\F{\mathbb{F}}
\def\P{\mathbb{P}}

\def\codeC{\mathscr{C}} 
\def\orb{\mathcal{O}rb} 

\title{Cliques in projective space and construction of cyclic Grassmannian Codes}

\author{Ismael Guti\'errez Garc\'ia}
\address{Department of Mathematics and Statistics, Universidad del Norte, Km 5 via a Puerto Colombia, Barranquilla - Colombia}
\email{isgutier@uninorte.edu.co}

\author{Ivan Molina Naizir}
\address{Department of Mathematics and Statistics, Universidad del Norte, Km 5 via a Puerto Colombia, Barranquilla - Colombia}
\email{inaizir@uninorte.edu.co}
\thanks{ }


\date{ }


\begin{document}
\maketitle

\begin{abstract}
The construction of Grassmannian codes in some projective space is of highly mathematical nature and requires strong computational power for the resulting searches. In this paper was constructed, using GAP System for Computational Discrete Algebra and Wolfram Mathematica, cliques in the projective space $\P_q(n)$ and then we use these to produce cyclic Grassmannian codes. 
\end{abstract}

\section{Introduction}
Cyclic Grassmannian codes were first presented by A. Kohnert and S. Kurz in \cite{Kurz} from the perspective of design theory  over finite fields. Later T. Etzion and A. Vardy in \cite{Etzion2} defined them as a $q$-analog of cyclic code from the classical coding theory. J. Rosenthal et al. \cite{Rosenthal} and H. Gluesing et al. \cite{Gluesing} studied cyclic codes from the point of view of groups actions. Specifically, they have used an action of the general linear group  over a Grassmannian  to define them: these codes were called cyclic orbits codes. Cyclic Grassmannian codes are a special case of orbits codes. 

Recently T. Etzion et al.  \cite{Etzion1}, K. Otal et al. \cite{Otal1}, B. Chen, and H. Liu \cite{Chen1} presented  new methods for constructing such codes, what includes linearized polynomials, namely subspace polynomials and  Frobenius mappings.  

A computational method for construction of cyclic Grassmannian codes was presented in \cite{Gutierrez-Molina}.

Let $\F_{q^n}$ be the extension field, of degree $n$, of the finite field with $q$ elements, $\F_q$ (where $q$ is a prime power). It is well known that we may regard $\F_{q^n}$ as a vector space of dimension $n$ over $\F_q$. That is, for a fixed basis, we can identify every element of $\F_{q^n}$ with a $n$-tuple of elements in $\F_q$. Therefore, we will not distinguish between $\F_{q^n}$ and $\F_q^n$. We denote with $\P_q(n)$ the projective space of order $n$, that is, the set of all subspaces of $\F_q^n$, including the null space and $\F_q^n$ itself. 

For a fixed natural number $k$, with $0\leq k\leq n$ we denote with $G_q(n,k)$ the set of all subspaces of $\F_q^n$ of dimension $k$ and we call it the $k$-Grassmannian over $\F_q$ or Grassmannian in short. We say that $\codeC\subseteq G_q(n,k)$ is an $(n,M,d,k)_q$ Grassmannian code if $|\codeC| = M$ and $d(X,Y)\geq d$ for all distinct $X,Y\in \codeC$. Such a code is also called a constant dimension code.

Let $\mathcal{A}_q(n,d,k)$ and $\mathcal{C}_q(n,d,k)$ be the maximum number of codewords in an $(n,M,d,k)_q$ grassmannian code over the filed $\F_q$ and the maximum number of codewords
in an $(n,M,d,k)_q$ cyclic code over $\F_q$, respectively. It is clear that $\mathcal{C}_q(n,d,k)\leq \mathcal{A}_q(n,d,k)$.

Let $\alpha\in \F_{q^n}^\ast$ and $V\in G_q(n,k)$. The \texttt{cyclic shift} of $V$ is defined as follows:
\[\alpha V := \{\alpha v\mid v\in V\}.\] 
Clearly $\alpha V$ is  a subspace belonging to $G_q(n,k)$. That is, it has the same dimension as $V$. A Grassmannian code $\codeC\subseteq G_q(n,k)$ is called \texttt{cyclic}, if for all $\alpha\in \F_{q^n}^\ast$ and all subspace $V\in \codeC$ we have that $\alpha V\in \codeC$. The set $\orb(V) :=\{\alpha V\mid \alpha\in \F_{q^n}^\ast\}$ is called the \texttt{orbit} of $V$.

Observe that in this definition the zero vector was omitted from the set of an orbit. Starting now, this will be explicitly deleted when we specify the elements of a codeword of a cyclic Grassmannian  code.

If $V\in G_q(n,k)$, then $|\orb(V)| = \tfrac{q^n-1}{q^t-1}$, 
for some natural number $t$, which divides $n$, see \cite[Lemma 9]{Etzion1}.

\begin{theorem}\label{theorem1}
$$\mathcal{C}_q(n,d,k)=\sum_{t\mid n}{\alpha_t\frac{q^n-1}{q^t-1}}$$
for some integer $0\leq\alpha_t$.
\end{theorem}

\section{Cliques construction}

A clique of an undirected graph $G$ is a complete subgraph of $G$; that is, A clique is a subset of vertices of $G$ such that every two distinct vertices in the clique are adjacent. The clique of the largest possible size is referred to as a maximum clique; that is, it cannot be extended by including one more adjacent vertex. The clique number $\omega(G)$ of G is the number of vertices in a maximum clique in $G$. A clique of size $k$ is called a $k$-clique.

To calculate the coefficients $\alpha_t$ in the previous theorem we proceed as follows:

\begin{enumerate}[(1)]
\item Find all the orbits of $ G_q (n, k) $ and denote this set by $\mathfrak{O}$. That is,
	\[\mathcal{O} := \{\orb(V) \mid V\in G_q(n,k)\}.\]
	
\item Calculate the minimum subspace distance $d_{\orb(\cdot)}$ of each orbit independently; then we form the pair  $(\orb(\cdot), d_{\orb(\cdot)})$.
	
\item A minimum distance $d$ is fixed, for which we want to obtain a cyclic code.
	
\item The graph $\mathcal{G} = (\mathcal{O,E})$ is constructed so that the set $\mathcal{E}$ of edges is obtained in the following way: two orbits are adjacent if their union has a minimum distance greater or equal than $d$. 

\item A clique in the graph $\mathcal{G}$ constructed in (4) is a Grassmannian cyclic code with minimum distance $d$ and dimension $k$.
	
\item To determine the maximum values of each $\alpha_t$, the graph $\mathcal{G}$ is separated into independent subgraphs by the number of spaces in their orbits (every vertex in each subgraph with the same number of associated spaces), and the number of cliques in each one is calculated.
\end{enumerate}

\begin{remark}
To perform the previous algorithm we use:
\begin{enumerate}[(1)]
\item GAP to calculate all the vector spaces over the field $\F_q$;
\item Java to construct the orbits  and graph $\mathcal{G}$;
\item Wolfram Mathematica to calculate the cliques. 
\end{enumerate}
\end{remark}

\newpage

\noindent\textbf{Algorithm 1.}

\begin{algorithm}[H]
\KwData{$d$: the minimum distance required for the code}
\KwResult{Grassmannian cyclic codes with minimum distance $d$.}
	
Let $V := \{O \subseteq G_q(n,k) \mid O \ \text{is an orbit}\}$\;
	$E \gets \{\}$\;
	\ForAll{$O_1 \in V$}{
		\ForAll{$O_2 \in V\setminus \{O_1\}$}{
			\If{ $d_{O_1}\geq d$ {\bf and} $d_{O_2}\geq d$ {\bf and} $D(O_1,O_2)\geq d$}{
				$E \gets E \cup \{(O_1,O_2)\}$\;
			}
		}
	}

36/5000
We define $G$ as the graph of orbits\;
	$G \gets (V,E)$\;
	\ForAll{C in Cliques(G)}{
		print(C)\;
	}
\caption{The algorithm that calculates all the cyclic codes of a Grassmannian }
\end{algorithm}

\vskip0.5cm

\noindent\textbf{Algorithm 2.}

\begin{algorithm}[H]
	\KwData{$n,d,k,q,t$ }
	\KwResult{A bound for $\alpha_t$}
Let $V := \{O \subseteq G_q(n,k) \mid O \ \text{is an orbit}$ {\bf and} $|O|= \frac{q^n-1}{q^t-1}\}$\;
	$E \gets \{\}$\;
	\ForAll{$O_1 \in V$}{
		\ForAll{$O_2 \in V\setminus \{O_1\}$}{
			\If{ $d_{O_1}\geq d$ {\bf and} $d_{O_2}\geq d$ {\bf and} $D(O_1,O_2)\geq d$}{
				$E \gets E \cup \{(O_1,O_2)\}$\;
			}
		}
	}
We define $G$ as the graph of orbits\;
$G \gets (V,E)$\;
	print("$\alpha_t\leq$",NumeroDeClique(G))\; 
\caption{Algorithm that calculates the upper bounds of the values of $\alpha_t$}
\end{algorithm}

\section{Classification of binary Grassmannian codes of length smaller than 6}

\begin{theorem}
$\mathcal{C}_2(4,4,2) = 5$.
\end{theorem}

\begin{proof}
Let $\alpha$ be a primitive root of $x^4+...$ and use this polynomial to generate the field $\F_{2^{4}}$. Let $\codeC \subseteq G_2(4,2)$  which consists of all cyclic shifts of
\[\{ \alpha^{0}, \alpha^{5}, \alpha^{10}\}.\]
This code  $\codeC$ is an $[4,5,4,2]$-cyclic code.  It consists of a unique orbit with 5 subspaces.
\end{proof}

\begin{theorem}
If $n<6$ then  $\mathcal{C}_2(n,4,k) = 0$.
\end{theorem}
 
\begin{proof}
The unique orbit with minimum distance 4 and $n<6$ was presented in the previous theorem.
\end{proof}

\section{Classification of binary Grassmannian codes of length 6}

 \subsection{Calculating the number  $\mathcal{C}_2(6,6,3)$}
It follows from  Theorem \ref{theorem1}  that 
\[\mathcal{C}_2(6,6,3)=63\alpha_1+21\alpha_2+9\alpha_3.\]

\begin{lemma}
Let $\codeC \subseteq G_2(6,3)$ a cyclic code with minimum distance 6. Then  
\begin{enumerate}[(1)]
\item $\alpha_1 \leq 0$
\item $\alpha_2 \leq 0$
\item $\alpha_3 \leq 1$
\end{enumerate}
\end{lemma}

\begin{proof}
There are not orbits with minimum distance 6 having 63 or 21 subspaces. There is a single orbit with minimum distance 6 and nine subspaces.  
\end{proof}

\begin{theorem}
 $\mathcal{C}_2(6,6,3) = 9$.
\end{theorem}

\begin{proof}
Let $\alpha$ be a primitive root of $x^6+...$ and use this polynomial to generate the field $\F_{2^{6}}$.  Let $\codeC \subseteq G_2(6,3)$   which consists of all cyclic shifts of
\[\{ \alpha^{0}, \alpha^{9}, \alpha^{18}, \alpha^{27}, \alpha^{36}, \alpha^{45}, \alpha^{54}\}.\]
This code  $\codeC$ is an $[6,9,6,3]$-cyclic code.  It consists of a unique orbit with nine subspaces.
\end{proof}

 \subsection{Calculating the number  $\mathcal{C}_2(6,4,3)$}
It follows from  Theorem \ref{theorem1}  that 
\[\mathcal{C}_2(6,4,3)=63\alpha_1+21\alpha_2+9\alpha_3.\]

\begin{lemma}
Let $\codeC \subseteq G_2(6,3)$ a cyclic code with minimum distance 4. Then  
\begin{enumerate}[(1)]
\item $\alpha_1 \leq 1$
\item $\alpha_2 \leq 0$
\item $\alpha_3 \leq 1$
\end{enumerate}
\end{lemma}

\begin{proof}
The constructed graph with these parameters is the null graph. That is an edge-less graph. Therefore the clique number is one. There are no orbits with 21 subspaces and minimum distance 4.
\end{proof}

\begin{lemma}
$\alpha_1+\alpha_3=1$
\end{lemma}

\begin{proof}
The combined constructed graph with the orbits with 63 and 21 subspaces in the null graph. Then the cyclic Grassmannian code has an orbit of 63 subspaces or an orbit of 9 subspaces but not both. 
\end{proof}

\begin{theorem}
 $\mathcal{C}_2(6,4,3) = 63$.
\end{theorem}

\begin{proof}
Let $\alpha$ be a primitive root of $x^6+...$ and use this polynomial to generate the field $\F_{2^6}$. Let $\codeC \subseteq G_2(6,3)$   which consists of all cyclic shifts of
\[\{ \alpha^{0}, \alpha^{6}, \alpha^{15}, \alpha^{26}, \alpha^{33}, \alpha^{34}, \alpha^{38}\}.\]
This code  $\codeC$ is an $[6,63,6,3]$-cyclic code.  It consists of a unique orbit with 63 subspaces.
\end{proof}

\subsection{Calculating the number  $\mathcal{C}_2(6,4,2)$}
It follows from  Theorem \ref{theorem1}  that 
\[\mathcal{C}_2(6,4,2)=63\alpha_1+21\alpha_2+9\alpha_3.\]

\begin{lemma}
Let $\codeC \subseteq G_2(6,2)$ a cyclic code with minimum distance 4. Then  
\begin{enumerate}[(1)]
\item $\alpha_1 \leq 0$
\item $\alpha_2 \leq 1$
\item $\alpha_3 \leq 0$
\end{enumerate}
\end{lemma}

\begin{proof}
There are not orbits with minimum distance 4 having 63 or 9 subspaces. The associate graph with the orbits of 21 subspaces is the null graph. 
\end{proof}

\begin{theorem}
 $\mathcal{C}_2(6,4,2) = 21$.
\end{theorem}

\begin{proof}
Let $\alpha$ be a primitive root of $x^6+...$ and use this polynomial to generate the field $\F_{2^{6}}$. Let $\codeC \subseteq G_2(6,2)$   which consists of all cyclic shifts of
\[\{ \alpha^{0}, \alpha^{21}, \alpha^{42}\}.\]
This code  $\codeC$ is an $[6,21,4,2]$-cyclic code.  It consists of a unique orbit with 21 subspaces.
\end{proof}

\begin{remark}
Results table for $n=6$ and $q=2$.
 \begin{center}
 \begin{table}[h]
$\begin{array}{|c|c|c|}
\hline
d\setminus k & 2 & 3 \\
\hline
4 & 21 & 63\\
\hline
6 & 0 & 9\\
\hline
\end{array}$\vskip0.5cm
 \caption{Tabla para $\mathcal{C}_2(6,d,k)$}
  \end{table}
 \end{center}
\end{remark}

\section{Classification of binary Grassmannian codes of length 7}

 \subsection{Calculating the number  $\mathcal{C}_2(7,4,3)$}

It follows from  Theorem \ref{theorem1}  that 
\[\mathcal{C}_2(7,4,3)=127\alpha_1.\]

\begin{figure}[h] 
\centering
\includegraphics[scale=0.22]{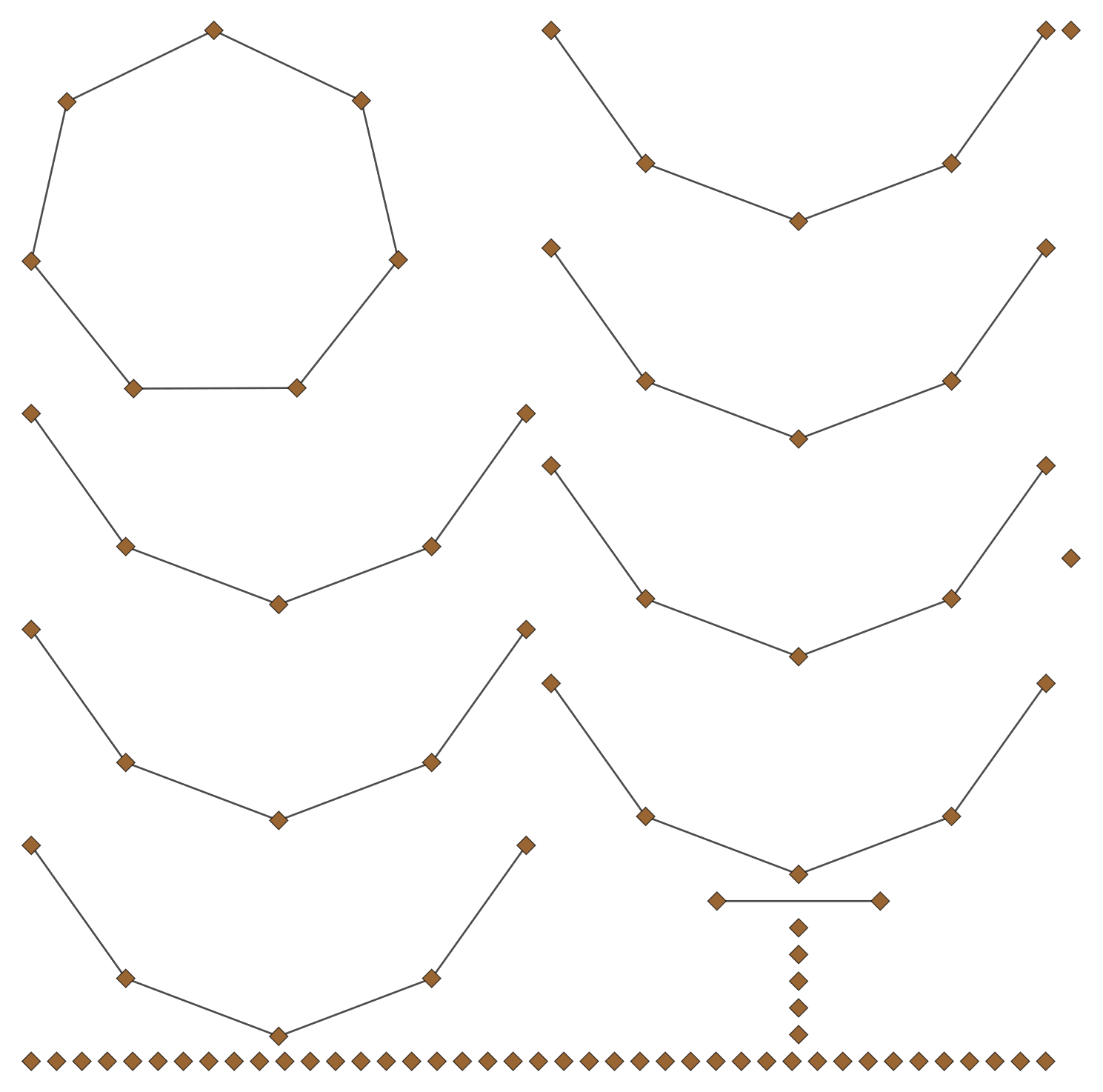} 
\caption{Cliques}\label{figura 1}
\end{figure}

\begin{lemma}
For a cyclic code  $\codeC \subseteq G_2 (7,3)$ with minimum distance 4 holds that  $\alpha_1 \leq 2$. 
\end{lemma}

\begin{proof}
The figure above illustrates this result. We can see that there are various cliques and the big one has two vertices, that is two orbits of 127 subspaces.
\end{proof}

\begin{theorem}
 $\mathcal{C}_2(7,4,3) = 254$.
\end{theorem}

\begin{proof}
Let $\alpha$ be a primitive root of $x^7+...$ and use this polynomial to generate the field $\F_{2^{7}}$.  Let $\codeC \subseteq G_2(7,3)$   which consists of all cyclic shifts of
\begin{eqnarray*}
\{ \alpha^{0}, \alpha^{4}, \alpha^{9}, \alpha^{28}, \alpha^{38}, \alpha^{58}, \alpha^{90}\}\\
\{ \alpha^{0}, \alpha^{8}, \alpha^{23}, \alpha^{39}, \alpha^{56}, \alpha^{82}, \alpha^{100}\}.
\end{eqnarray*}
This code  $\codeC$ is an $[8, 254,4,3]$-cyclic code.  It consists of two orbits with 127 subspaces.
\end{proof}

\begin{remark}
Results table for $n=7$ and $q=2$.
 \begin{center}
 \begin{table}[h]
$\begin{array}{|c|c|c|}
\hline
d\setminus k & 2 & 3 \\
\hline
4 & 0 & 254\\
\hline
6 & 0 & 0\\
\hline
\end{array}$\vskip0.5cm
 \caption{ $\mathcal{C}_2(7,d,k)$}
  \end{table}
 \end{center}
\end{remark}

\section{Classification of binary Grassmannian codes of length 8}

\subsection{Calculating the number $\mathcal{C}_2(8,4,4)$}

It follows from  Theorem \ref{theorem1}  that 
\[\mathcal{C}_2(8,4,4)=255\alpha_1+85\alpha_2+17\alpha_4.\]

\begin{lemma}\label{6.1}
Let $\codeC \subseteq G_2(8,4)$ a cyclic code with minimum distance 4. Then 
\begin{enumerate}[(1)]
\item $\alpha_1 \leq 17$;
\item $\alpha_2 \leq 4$;
\item $\alpha_4 \leq 1$.
\end{enumerate}
\end{lemma}

\begin{proof}
The graph formed only by the orbits of 255 subspaces has a clique of 17 vertices, and there is no clique of greater size. In the same way, the graph formed by the orbits of 85 subspaces and the orbits of 17 subspaces is the null graph.
\end{proof}

\begin{lemma}\label{6.2}
If $\alpha_4=1$ then $\alpha_1+\alpha_2\leq 3$
\end{lemma}

\begin{proof}
Fixing the orbit of 17 subspaces in all cliques, then the combined graph formed by the orbits of 255 subspaces and 85 subspaces do not contain a clique with more than four vertices. 
\end{proof}

\begin{theorem}
	$\mathcal{C}_2(8,4,4) \leq 4675 $.
\end{theorem}

\begin{proof}
It follows directly from the two previous lemmas. 
\end{proof}

\begin{lemma}\label{6.4}
If $\alpha_2=4$ then $\alpha_1\leq 16$.
\end{lemma}

\begin{proof}
Similar to the previous theorem, but now we fix the clique of four orbits with 85 subspaces. This procedure is made for every combination of four orbits of 85 subspaces that form a clique.
\end{proof}

\begin{theorem}
There is a cyclic code with 4420 codewords. That is, $\alpha_2=4 $ and $\alpha_1=16$.  
\end{theorem}

\begin{proof}
Let $\alpha$ be a primitive root of $x^8+x^4+x^3+x^2+1$ and use this polynomial to generate the field $\F_{2^{8}}$. Let $\codeC \subseteq G_2(8,4)$ which consists of all cyclic shifts of

\begin{align*}
	\{ \alpha^{0}, \alpha^{7}, \alpha^{30}, \alpha^{46}, \alpha^{66}, \alpha^{76}, \alpha^{87}, \alpha^{88}, \alpha^{89}, \alpha^{112}, \alpha^{113}, \alpha^{137}, \alpha^{167}, \alpha^{175}, \alpha^{203}\}\\
	\{ \alpha^{0}, \alpha^{40}, \alpha^{41}, \alpha^{53}, \alpha^{65}, \alpha^{80}, \alpha^{84}, \alpha^{98}, \alpha^{124}, \alpha^{139}, \alpha^{147}, \alpha^{157}, \alpha^{162}, \alpha^{168}, \alpha^{180}\}\\
	\{ \alpha^{0}, \alpha^{2}, \alpha^{31}, \alpha^{45}, \alpha^{50}, \alpha^{91}, \alpha^{110}, \alpha^{123}, \alpha^{126}, \alpha^{163}, \alpha^{182}, \alpha^{183}, \alpha^{205}, \alpha^{207}, \alpha^{209}\}\\
	\{ \alpha^{0}, \alpha^{27}, \alpha^{59}, \alpha^{62}, \alpha^{82}, \alpha^{89}, \alpha^{90}, \alpha^{104}, \alpha^{114}, \alpha^{117}, \alpha^{122}, \alpha^{125}, \alpha^{166}, \alpha^{194}, \alpha^{203}\}\\
	\{ \alpha^{0}, \alpha^{1}, \alpha^{25}, \alpha^{56}, \alpha^{64}, \alpha^{65}, \alpha^{70}, \alpha^{71}, \alpha^{89}, \alpha^{95}, \alpha^{109}, \alpha^{131}, \alpha^{162}, \alpha^{176}, \alpha^{203}\}\\
	\{ \alpha^{0}, \alpha^{1}, \alpha^{25}, \alpha^{38}, \alpha^{81}, \alpha^{94}, \alpha^{124}, \alpha^{155}, \alpha^{156}, \alpha^{159}, \alpha^{160}, \alpha^{169}, \alpha^{180}, \alpha^{184}, \alpha^{202}\}\\
	\{ \alpha^{0}, \alpha^{7}, \alpha^{9}, \alpha^{57}, \alpha^{62}, \alpha^{64}, \alpha^{70}, \alpha^{72}, \alpha^{83}, \alpha^{90}, \alpha^{112}, \alpha^{120}, \alpha^{156}, \alpha^{169}, \alpha^{195}\}\\
	\{ \alpha^{0}, \alpha^{8}, \alpha^{16}, \alpha^{54}, \alpha^{69}, \alpha^{87}, \alpha^{125}, \alpha^{130}, \alpha^{145}, \alpha^{163}, \alpha^{167}, \alpha^{182}, \alpha^{194}, \alpha^{200}, \alpha^{208}\}\\
	\{ \alpha^{0}, \alpha^{5}, \alpha^{10}, \alpha^{21}, \alpha^{37}, \alpha^{40}, \alpha^{76}, \alpha^{84}, \alpha^{113}, \alpha^{114}, \alpha^{138}, \alpha^{143}, \alpha^{150}, \alpha^{166}, \alpha^{179}\}\\
	\{ \alpha^{0}, \alpha^{23}, \alpha^{64}, \alpha^{70}, \alpha^{79}, \alpha^{97}, \alpha^{110}, \alpha^{124}, \alpha^{126}, \alpha^{154}, \alpha^{174}, \alpha^{180}, \alpha^{190}, \alpha^{196}, \alpha^{201}\}\\
	\{ \alpha^{0}, \alpha^{16}, \alpha^{31}, \alpha^{45}, \alpha^{49}, \alpha^{88}, \alpha^{114}, \alpha^{145}, \alpha^{155}, \alpha^{159}, \alpha^{166}, \alpha^{171}, \alpha^{175}, \alpha^{197}, \alpha^{211}\}\\
	\{ \alpha^{0}, \alpha^{19}, \alpha^{47}, \alpha^{62}, \alpha^{78}, \alpha^{80}, \alpha^{90}, \alpha^{92}, \alpha^{101}, \alpha^{128}, \alpha^{140}, \alpha^{168}, \alpha^{205}, \alpha^{207}, \alpha^{212}\}\\
	\{ \alpha^{0}, \alpha^{2}, \alpha^{29}, \alpha^{39}, \alpha^{49}, \alpha^{50}, \alpha^{60}, \alpha^{71}, \alpha^{74}, \alpha^{103}, \alpha^{106}, \alpha^{109}, \alpha^{132}, \alpha^{181}, \alpha^{197}\}\\
	\{ \alpha^{0}, \alpha^{9}, \alpha^{28}, \alpha^{38}, \alpha^{47}, \alpha^{49}, \alpha^{93}, \alpha^{97}, \alpha^{101}, \alpha^{120}, \alpha^{158}, \alpha^{184}, \alpha^{190}, \alpha^{193}, \alpha^{197}\}\\
	\{ \alpha^{0}, \alpha^{7}, \alpha^{47}, \alpha^{59}, \alpha^{79}, \alpha^{82}, \alpha^{91}, \alpha^{94}, \alpha^{101}, \alpha^{112}, \alpha^{148}, \alpha^{174}, \alpha^{202}, \alpha^{206}, \alpha^{209}\}\\
	\{ \alpha^{0}, \alpha^{6}, \alpha^{12}, \alpha^{49}, \alpha^{53}, \alpha^{58}, \alpha^{107}, \alpha^{127}, \alpha^{147}, \alpha^{149}, \alpha^{156}, \alpha^{169}, \alpha^{188}, \alpha^{191}, \alpha^{197}\}\\
	\{ \alpha^{0}, \alpha^{7}, \alpha^{19}, \alpha^{27}, \alpha^{49}, \alpha^{85}, \alpha^{92}, \alpha^{104}, \alpha^{112}, \alpha^{134}, \alpha^{170}, \alpha^{177}, \alpha^{189}, \alpha^{197}, \alpha^{219}\}\\
	\{ \alpha^{0}, \alpha^{6}, \alpha^{10}, \alpha^{21}, \alpha^{39}, \alpha^{85}, \alpha^{91}, \alpha^{95}, \alpha^{106}, \alpha^{124}, \alpha^{170}, \alpha^{176}, \alpha^{180}, \alpha^{191}, \alpha^{209}\}\\
	\{ \alpha^{0}, \alpha^{13}, \alpha^{14}, \alpha^{38}, \alpha^{54}, \alpha^{85}, \alpha^{98}, \alpha^{99}, \alpha^{123}, \alpha^{139}, \alpha^{170}, \alpha^{183}, \alpha^{184}, \alpha^{208}, \alpha^{224}\}\\
	\{ \alpha^{0}, \alpha^{9}, \alpha^{32}, \alpha^{35}, \alpha^{37}, \alpha^{85}, \alpha^{94}, \alpha^{117}, \alpha^{120}, \alpha^{122}, \alpha^{170}, \alpha^{179}, \alpha^{202}, \alpha^{205}, \alpha^{207}\}
\end{align*}

This code  $\codeC$ is an $[8,4420,4,4]$-cyclic code.
The first sixteen orbits are sets with 255 subspaces and the remaining four with 85 subspaces.
\end{proof}

\begin{theorem}
$\mathcal{C}_2(8,4,4) \geq 4420$.
\end{theorem}

\begin{proof}
The previous theorem show a cyclic Grassmannian code with 4420 subspaces.
\end{proof}

\begin{theorem}
There is a cyclic code with 4590 codewords. That is, $\alpha_2=3$ and $\alpha_1=17$. 
\end{theorem}

\begin{proof}
Let $\alpha$ be a primitive root of $x^8+x^4+x^3+x^2+1$ and use this polynomial to generate the field $\F_{2^{8}}$. Let $\codeC \subseteq G_2(8,4)$ which consists of all cyclic shifts of
\begin{align*}
\{ \alpha^{0}, \alpha^{7}, \alpha^{30}, \alpha^{46}, \alpha^{66}, \alpha^{76}, \alpha^{87}, \alpha^{88}, \alpha^{89}, \alpha^{112}, \alpha^{113}, \alpha^{137}, \alpha^{167}, \alpha^{175}, \alpha^{203}\}\\
\{ \alpha^{0}, \alpha^{40}, \alpha^{41}, \alpha^{53}, \alpha^{65}, \alpha^{80}, \alpha^{84}, \alpha^{98}, \alpha^{124}, \alpha^{139}, \alpha^{147}, \alpha^{157}, \alpha^{162}, \alpha^{168}, \alpha^{180}\}\\
\{ \alpha^{0}, \alpha^{2}, \alpha^{31}, \alpha^{45}, \alpha^{50}, \alpha^{91}, \alpha^{110}, \alpha^{123}, \alpha^{126}, \alpha^{163}, \alpha^{182}, \alpha^{183}, \alpha^{205}, \alpha^{207}, \alpha^{209}\}\\
\{ \alpha^{0}, \alpha^{27}, \alpha^{59}, \alpha^{62}, \alpha^{82}, \alpha^{89}, \alpha^{90}, \alpha^{104}, \alpha^{114}, \alpha^{117}, \alpha^{122}, \alpha^{125}, \alpha^{166}, \alpha^{194}, \alpha^{203}\}\\
\{ \alpha^{0}, \alpha^{1}, \alpha^{25}, \alpha^{56}, \alpha^{64}, \alpha^{65}, \alpha^{70}, \alpha^{71}, \alpha^{89}, \alpha^{95}, \alpha^{109}, \alpha^{131}, \alpha^{162}, \alpha^{176}, \alpha^{203}\}\\
\{ \alpha^{0}, \alpha^{1}, \alpha^{25}, \alpha^{38}, \alpha^{81}, \alpha^{94}, \alpha^{124}, \alpha^{155}, \alpha^{156}, \alpha^{159}, \alpha^{160}, \alpha^{169}, \alpha^{180}, \alpha^{184}, \alpha^{202}\}\\
\{ \alpha^{0}, \alpha^{7}, \alpha^{9}, \alpha^{57}, \alpha^{62}, \alpha^{64}, \alpha^{70}, \alpha^{72}, \alpha^{83}, \alpha^{90}, \alpha^{112}, \alpha^{120}, \alpha^{156}, \alpha^{169}, \alpha^{195}\}\\
\{ \alpha^{0}, \alpha^{8}, \alpha^{16}, \alpha^{54}, \alpha^{69}, \alpha^{87}, \alpha^{125}, \alpha^{130}, \alpha^{145}, \alpha^{163}, \alpha^{167}, \alpha^{182}, \alpha^{194}, \alpha^{200}, \alpha^{208}\}\\
\{ \alpha^{0}, \alpha^{5}, \alpha^{10}, \alpha^{21}, \alpha^{37}, \alpha^{40}, \alpha^{76}, \alpha^{84}, \alpha^{113}, \alpha^{114}, \alpha^{138}, \alpha^{143}, \alpha^{150}, \alpha^{166}, \alpha^{179}\}\\
\{ \alpha^{0}, \alpha^{23}, \alpha^{64}, \alpha^{70}, \alpha^{79}, \alpha^{97}, \alpha^{110}, \alpha^{124}, \alpha^{126}, \alpha^{154}, \alpha^{174}, \alpha^{180}, \alpha^{190}, \alpha^{196}, \alpha^{201}\}\\
\{ \alpha^{0}, \alpha^{16}, \alpha^{31}, \alpha^{45}, \alpha^{49}, \alpha^{88}, \alpha^{114}, \alpha^{145}, \alpha^{155}, \alpha^{159}, \alpha^{166}, \alpha^{171}, \alpha^{175}, \alpha^{197}, \alpha^{211}\}\\
\{ \alpha^{0}, \alpha^{19}, \alpha^{47}, \alpha^{62}, \alpha^{78}, \alpha^{80}, \alpha^{90}, \alpha^{92}, \alpha^{101}, \alpha^{128}, \alpha^{140}, \alpha^{168}, \alpha^{205}, \alpha^{207}, \alpha^{212}\}\\
\{ \alpha^{0}, \alpha^{2}, \alpha^{29}, \alpha^{39}, \alpha^{49}, \alpha^{50}, \alpha^{60}, \alpha^{71}, \alpha^{74}, \alpha^{103}, \alpha^{106}, \alpha^{109}, \alpha^{132}, \alpha^{181}, \alpha^{197}\}\\
\{ \alpha^{0}, \alpha^{9}, \alpha^{28}, \alpha^{38}, \alpha^{47}, \alpha^{49}, \alpha^{93}, \alpha^{97}, \alpha^{101}, \alpha^{120}, \alpha^{158}, \alpha^{184}, \alpha^{190}, \alpha^{193}, \alpha^{197}\}\\
\{ \alpha^{0}, \alpha^{7}, \alpha^{47}, \alpha^{59}, \alpha^{79}, \alpha^{82}, \alpha^{91}, \alpha^{94}, \alpha^{101}, \alpha^{112}, \alpha^{148}, \alpha^{174}, \alpha^{202}, \alpha^{206}, \alpha^{209}\}\\
\{ \alpha^{0}, \alpha^{6}, \alpha^{12}, \alpha^{49}, \alpha^{53}, \alpha^{58}, \alpha^{107}, \alpha^{127}, \alpha^{147}, \alpha^{149}, \alpha^{156}, \alpha^{169}, \alpha^{188}, \alpha^{191}, \alpha^{197}\}\\
\{ \alpha^{0}, \alpha^{4}, \alpha^{30}, \alpha^{32}, \alpha^{35}, \alpha^{49}, \alpha^{66}, \alpha^{80}, \alpha^{94}, \alpha^{100}, \alpha^{117}, \alpha^{122}, \alpha^{168}, \alpha^{197}, \alpha^{202}\}\\
\{ \alpha^{0}, \alpha^{7}, \alpha^{19}, \alpha^{27}, \alpha^{49}, \alpha^{85}, \alpha^{92}, \alpha^{104}, \alpha^{112}, \alpha^{134}, \alpha^{170}, \alpha^{177}, \alpha^{189}, \alpha^{197}, \alpha^{219}\}\\
\{ \alpha^{0}, \alpha^{6}, \alpha^{10}, \alpha^{21}, \alpha^{39}, \alpha^{85}, \alpha^{91}, \alpha^{95}, \alpha^{106}, \alpha^{124}, \alpha^{170}, \alpha^{176}, \alpha^{180}, \alpha^{191}, \alpha^{209}\}\\
\{ \alpha^{0}, \alpha^{13}, \alpha^{14}, \alpha^{38}, \alpha^{54}, \alpha^{85}, \alpha^{98}, \alpha^{99}, \alpha^{123}, \alpha^{139}, \alpha^{170}, \alpha^{183}, \alpha^{184}, \alpha^{208}, \alpha^{224}\}
\end{align*}

This code  $\codeC$ is an $[8,4590,4,4]$-cyclic code.
The first seventeen orbits are sets with 255 subspaces and the remaining three with 85 subspaces.
\end{proof}

\begin{theorem}
 $\mathcal{C}_2(8,4,4) = 4590$.
\end{theorem}

\begin{proof}
It follows directly from previous theorem and lemmas \ref{6.1}, \ref{6.2} and \ref{6.4}. 
\end{proof}

\subsection{Calculating the number $\mathcal{C}_2(8,4,3)$}

It follows from  Theorem \ref{theorem1}  that 
\[\mathcal{C}_2(8,4,3)=255\alpha_1+85\alpha_2+17\alpha_4.\]

\begin{lemma}
Let $\codeC \subseteq G_2(8,3)$ a cyclic code with minimum distance 4. Then 
\begin{enumerate}[(1)]
\item $\alpha_1 \leq 5$;
\item $\alpha_2 \leq 0$;
\item $\alpha_4 \leq 0$.
\end{enumerate}
\end{lemma}

\begin{proof}
There are not orbits with 85 and 17 subspaces. There is a clique with five orbits, and there is not one with six orbits.
\end{proof}

\begin{theorem}
There is a cyclic code with 1275 codewords. That is,  $\alpha_1=5$.
\end{theorem}

\begin{proof}
Let $\alpha$ be a primitive root of $x^8+x^4+x^3+x^2+1$ and use this polynomial to generate the field $\F_{2^{8}}$. Let $\codeC \subseteq G_2(8,3)$ which consists of all cyclic shifts of
\begin{align*}
\{ \alpha^{0}, \alpha^{27}, \alpha^{34}, \alpha^{98}, \alpha^{104}, \alpha^{136}, \alpha^{139}\}\\
\{ \alpha^{0}, \alpha^{58}, \alpha^{60}, \alpha^{107}, \alpha^{108}, \alpha^{132}, \alpha^{161}\}\\
\{ \alpha^{0}, \alpha^{76}, \alpha^{80}, \alpha^{95}, \alpha^{113}, \alpha^{168}, \alpha^{176}\}\\
\{ \alpha^{0}, \alpha^{20}, \alpha^{42}, \alpha^{59}, \alpha^{82}, \alpha^{110}, \alpha^{126}\}\\
\{ \alpha^{0}, \alpha^{13}, \alpha^{69}, \alpha^{99}, \alpha^{130}, \alpha^{135}, \alpha^{144}\}
\end{align*}
This code  $\codeC$ is an $[8,1275,4,3]$-cyclic code. Every orbit has 255 subspaces. 
\end{proof}

\begin{theorem}
 $\mathcal{C}_2(8,4,3) = 1275$.
\end{theorem}

\begin{proof}
See the previous theorem.
\end{proof}

\subsection{Calculating the number $\mathcal{C}_2(8,6,4)$}

\begin{theorem}
 $\mathcal{C}_2(8,6,4) = 0$
\end{theorem}

\begin{proof}
There is no orbit with minimum distance 6.
\end{proof}

\subsection{Calculating the number  $\mathcal{C}_2(8,8,4)$}

It follows from  Theorem \ref{theorem1}  that 
\[\mathcal{C}_2(8,8,4)=255\alpha_1+85\alpha_2+17\alpha_4.\]

\begin{lemma}
Let $\codeC \subseteq G_2(8,4)$ a cyclic code with minimum distance 8. Then  
\begin{enumerate}[(1)]
\item $\alpha_1 \leq 0$;
\item $\alpha_2 \leq 0$;
\item $\alpha_4 \leq 1$.
\end{enumerate}
\end{lemma}

\begin{proof}
There is a unique orbit with a minimum distance of 8 and 17 subspaces.
\end{proof}

\begin{theorem}
There is a cyclic code with 17 codewords. That is,  $\alpha_4=1$.
\end{theorem}

\begin{proof}
Let $\alpha$ be a primitive root of $x^8+x^4+x^3+x^2+1$ and use this polynomial to generate the field $\F_{2^{8}}$. Let $\codeC \subseteq G_2(8,4)$   which consists of all cyclic shifts of
\[\{ \alpha^{0}, \alpha^{17}, \alpha^{34}, \alpha^{51}, \alpha^{68}, \alpha^{85}, \alpha^{102}, \alpha^{119}, \alpha^{136}, \alpha^{153}, \alpha^{170}, \alpha^{187}, \alpha^{204}, \alpha^{221}, \alpha^{238}\}.\]
This code  $\codeC$ is an $[8,17,8,4]$-cyclic code.  It consists of a unique orbit.
\end{proof}

\begin{theorem}
 $\mathcal{C}_2(8,8,4) = 17$.
\end{theorem}

\subsection{Calculating the number $\mathcal{C}_2(8,6,3)$}

\begin{theorem}
 $\mathcal{C}_2(8,6,3) = 0$.
\end{theorem}

\begin{proof}
There is no orbit with minimum distance of 6.
\end{proof}

\subsection{Calculating the number $\mathcal{C}_2(8,8,3)$}

\begin{theorem}
 $\mathcal{C}_2(8,8,3) = 0$.
\end{theorem}

\begin{proof}
There is no orbit with minimum distance 8.
\end{proof}

\subsection{Calculating the number  $\mathcal{C}_2(8,4,2)$}
It follows from  Theorem \ref{theorem1}  that 
\[\mathcal{C}_2(8,4,2)=255\alpha_1+85\alpha_2+17\alpha_4.\]

\begin{theorem}
Let $\codeC \subseteq G_2(8,2)$ a cyclic code with minimum distance of 4. Then  
 
\begin{enumerate}[(1)]
\item $\alpha_1 \leq 0$;
\item $\alpha_2 \leq 1$;
\item $\alpha_4 \leq 0$.
\end{enumerate}
\end{theorem}

\begin{proof}
There is only one orbit with 85 subspaces, and it has a minimum distance of 4. There are no orbits of other sizes with minimum distance 4.
\end{proof}

\begin{theorem}
There is a cyclic code with 85 codewords. That is,  $\alpha_2=1$.
\end{theorem}

\begin{proof}
Let $\alpha$ be a primitive root of $x^8+x^4+x^3+x^2+1$ and use this polynomial to generate the field $\F_{2^{8}}$. Let $\codeC \subseteq G_2(8,2)$  which consists of all cyclic shifts of
\[\{ \alpha^{0}, \alpha^{85}, \alpha^{170}\}.\]
This code  $\codeC$ is an $[8,85,4,2]$-cyclic code.  It consists of a unique orbit.
\end{proof}

\begin{theorem}
 $\mathcal{C}_2(8,4,2) = 85$.
\end{theorem}

\begin{remark}
Results table for $n=8$ and $q=2$.

 \begin{center}
 \begin{table}
 $\begin{array}{|c|c|c|c|}
 \hline
 d\setminus k & 2 & 3 & 4\\
 \hline
 4 & 85 & 1275 & 4590\\
 \hline
 6 & 0 & 0 & 0\\
 \hline
 8 & 0 & 0 & 17\\ 
 \hline
 \end{array}$\vskip0.5cm
 \caption{Tabla para $\mathcal{C}_2(8,d,k)$}
  \end{table}
 \end{center}
\end{remark}

\section{Classification of binary Grassmannian codes of length 9}

\subsection{Calculating the number  $\mathcal{C}_2(9,4,3)$}
It follows from  Theorem \ref{theorem1}  that 
\[\mathcal{C}_2(9,4,3)=511\alpha_1+73\alpha_3.\]

\begin{theorem}
Let $\codeC \subseteq G_2(9,3)$ a cyclic code with minimum distance 4. Then  
\begin{enumerate}[(1)]
\item $\alpha_1 \leq 11$;
\item $\alpha_3 \leq 1$.
\end{enumerate}
\end{theorem}

\begin{proof}
There is only one orbit with 73 subspaces, and it has a minimum distance of 6. There is a clique of 11 orbits with 511 subspaces.
\end{proof}

 \begin{theorem}
 $\mathcal{C}_2(9,4,3) \leq 5621$
 \end{theorem}
 
\begin{proof}
It follows directly from previous theorem.
\end{proof}
 
\begin{theorem}
There is a cyclic code with 5621 codewords. That is,  $\alpha_1=11$.
\end{theorem}

\begin{proof}
Let $\alpha$ be a primitive root of $x^9+...$ and use this polynomial to generate the field $\F_{2^{9}}$. Let $\codeC \subseteq G_2(9,3)$   which consists of all cyclic shifts of
\begin{align*}
\{ \alpha^{0}, \alpha^{26}, \alpha^{27}, \alpha^{142}, \alpha^{156}, \alpha^{276}, \alpha^{345}\}\\
\{ \alpha^{0}, \alpha^{86}, \alpha^{162}, \alpha^{169}, \alpha^{229}, \alpha^{237}, \alpha^{247}\}\\
\{ \alpha^{0}, \alpha^{33}, \alpha^{81}, \alpha^{110}, \alpha^{181}, \alpha^{305}, \alpha^{379}\}\\
\{ \alpha^{0}, \alpha^{2}, \alpha^{93}, \alpha^{96}, \alpha^{154}, \alpha^{260}, \alpha^{304}\}\\
\{ \alpha^{0}, \alpha^{28}, \alpha^{127}, \alpha^{232}, \alpha^{248}, \alpha^{268}, \alpha^{311}\}\\
\{ \alpha^{0}, \alpha^{25}, \alpha^{56}, \alpha^{90}, \alpha^{109}, \alpha^{227}, \alpha^{281}\}\\
\{ \alpha^{0}, \alpha^{133}, \alpha^{174}, \alpha^{185}, \alpha^{197}, \alpha^{277}, \alpha^{332}\}\\
\{ \alpha^{0}, \alpha^{21}, \alpha^{157}, \alpha^{194}, \alpha^{244}, \alpha^{306}, \alpha^{372}\}\\
\{ \alpha^{0}, \alpha^{73}, \alpha^{170}, \alpha^{187}, \alpha^{219}, \alpha^{259}, \alpha^{289}\}\\
\{ \alpha^{0}, \alpha^{35}, \alpha^{123}, \alpha^{180}, \alpha^{218}, \alpha^{231}, \alpha^{356}\}\\
\{ \alpha^{0}, \alpha^{24}, \alpha^{131}, \alpha^{177}, \alpha^{290}, \alpha^{294}, \alpha^{299}\}
\end{align*}
This code  $\codeC$ is an $[8,5621,4,3]$-cyclic code.  It consists of eleven orbits with 511 subspaces.
\end{proof}

\begin{theorem}
There is a cyclic code with 5694 codewords. That is,  $\alpha_1=11, \alpha_3=1$.
\end{theorem}

\begin{proof}
Let $\alpha$ be a primitive root of $x^9+...$ and use this polynomial to generate the field $\F_{2^{9}}$. Let $\codeC \subseteq G_2(9,3)$  which consists of all cyclic shifts of

\begin{align*}
\{ \alpha^{0}, \alpha^{64}, \alpha^{144}, \alpha^{242}, \alpha^{313}, \alpha^{381}, \alpha^{382}\}\\
\{ \alpha^{0}, \alpha^{23}, \alpha^{84}, \alpha^{190}, \alpha^{202}, \alpha^{335}, \alpha^{337}\}\\
\{ \alpha^{0}, \alpha^{27}, \alpha^{41}, \alpha^{52}, \alpha^{83}, \alpha^{142}, \alpha^{161}\}\\
\{ \alpha^{0}, \alpha^{15}, \alpha^{91}, \alpha^{94}, \alpha^{126}, \alpha^{166}, \alpha^{322}\}\\
\{ \alpha^{0}, \alpha^{45}, \alpha^{53}, \alpha^{63}, \alpha^{225}, \alpha^{263}, \alpha^{310}\}\\
\{ \alpha^{0}, \alpha^{26}, \alpha^{33}, \alpha^{93}, \alpha^{181}, \alpha^{276}, \alpha^{304}\}\\
\{ \alpha^{0}, \alpha^{102}, \alpha^{141}, \alpha^{184}, \alpha^{206}, \alpha^{316}, \alpha^{397}\}\\
\{ \alpha^{0}, \alpha^{57}, \alpha^{108}, \alpha^{157}, \alpha^{227}, \alpha^{244}, \alpha^{281}\}\\
\{ \alpha^{0}, \alpha^{74}, \alpha^{223}, \alpha^{239}, \alpha^{259}, \alpha^{289}, \alpha^{351}\}\\
\{ \alpha^{0}, \alpha^{6}, \alpha^{103}, \alpha^{158}, \alpha^{192}, \alpha^{308}, \alpha^{329}\}\\
\{ \alpha^{0}, \alpha^{24}, \alpha^{131}, \alpha^{177}, \alpha^{290}, \alpha^{294}, \alpha^{299}\}\\
\{ \alpha^{0}, \alpha^{73}, \alpha^{146}, \alpha^{219}, \alpha^{292}, \alpha^{365}, \alpha^{438}\}
\end{align*}

This code  $\codeC$ is an $[8,5694,4,3]$-cyclic code.  It consists of eleven orbits with 511 subspaces and the remaining orbit with 73 subspaces.
\end{proof}

\subsection{Calculating the number  $\mathcal{C}_2(9,4,2)$}

\begin{theorem}
 $\mathcal{C}_2(9,4,2) = 0$.
\end{theorem}

\begin{proof}
There is no orbit with minimum distance 4.
\end{proof}

\begin{remark}
Results table for $n=9$ and $q=2$.
 \begin{center}
 \begin{table}[h]
$\begin{array}{|c|c|c|c|}
\hline
d\setminus k & 2 & 3 & 4\\
\hline
4 & 0 & 5694 & 25500-50589 \\
\hline
6 & 0 & 73 & 511 \\
\hline
8 & 0 & 0 & 0\\ 
\hline
\end{array}$\vskip0.5cm
 \caption{Tabla para $\mathcal{C}_2(9,d,k)$}
  \end{table}
 \end{center}
\end{remark}


\begin{thebibliography}{99}
	
\bibitem{Etzion1} E. Ben-Sasson, T. Etzion, A. Gabizon and N. Raviv. Subspace Polynomials and Cyclic Subspace Codes. arXiv:1404.7739, 2015.

\bibitem{Chen1}  B. Chen, and H. Liu.  Constructions of cyclic constant dimension codes.  Designs, Codes and Cryptography, Vol 86, pages 1267 - 1279,  2018.

\bibitem{Etzion2} T. Etzion and A. Vardy. Error-Correcting Codes in Projective Space. IEEE Transactions on Information Theory, Volume 57 (2), pages 1165-1173, 2011.

\bibitem{Gluesing}  H. Gluesing-Luerssen, K. Morrison, and C. Troha.  Cyclic orbit codes and stabilizer subfields. Advances in Mathematics of Communications, vol 9, pages 177-197,  2015.

\bibitem{Gutierrez-Molina}  I. Gutierrez, and I. Molina.  Some constructions of cyclic and quasi-cyclic subspaces codes.  arXiv:1504.04553v4, 2015.

	
\bibitem{Kurz}  A. Kohnert, and S. Kurz.  Construction of Large Constant Dimension Codes with a Prescribed Minimum Distance.  Mathematical Methods in Computer Science. Lecture Notes in Computer Science, vol. 5393,  pages 31-42, 2008.


\bibitem{Otal1} K. Otal, and F. \"Ozbudak. Cyclic subspace codes via subspace polynomials.  Designs, Codes and Cryptography.  doi:10.1007/s10623-016-0297-1.  2016.

\bibitem{Rosenthal} A.-L. Trautmann, F. Manganiello, M. Braun, and J. Rosenthal. Cyclic orbit codes. IEEE Transactions on Information Theory. Vol 59, pages 7386-7404, 2013.	


\end{thebibliography}
\end{document}